\documentclass[a4paper,UKenglish,cleveref, autoref, thm-restate]{lipics-v2019}

\usepackage{algorithm,algorithmic, comment}
\nolinenumbers

\title{Improved Approximation Algorithms for \\Weighted Edge Coloring of Graphs} 

\titlerunning{} 

\author{Debarsho Sannyasi}{Indian Institute of Technology, Kanpur, India \and \url{debarsho@iitk.ac.in} }{}{}{}

\authorrunning{Debarsho Sannyasi} 

\Copyright{Debarsho Sannyasi} 



\category{} 

\relatedversion{} 

\supplement{}

\ArticleNo{1}

\begin{document}

\maketitle

\begin{abstract}
We study \emph{weighted edge coloring of graphs}, 
where we are given an undirected edge-weighted general multi-graph $G := (V, E)$ with weights $w : E \rightarrow [0, 1]$. The goal is to find a proper weighted coloring
of the edges with as few colors as possible. An edge coloring  is called a
proper weighted coloring if the sum of the weights of the edges incident to a vertex of any color
is at most one. 

In the online setting, the edges are revealed one by one and have to be colored irrevocably as soon as they are revealed.  We show that $3.39m+o(m)$ colors are enough when the maximum number of neighbors of a vertex over all the vertices is $o(m)$ and where $m$ is the maximum over all vertices of the minimum number of unit-sized bins needed to pack the weights of the incident edges to that vertex. We also prove the tightness of our analysis. This improves upon the previous best upper bound of $5m$ by Correa and Goemans [STOC 2004]. 

For the offline case, we show that for a simple graph with edge disjoint cycles, $m+1$ colors are sufficient and for a multi-graph tree, we show that $1.693m+12$ colors are sufficient.   
\end{abstract}

\section{Introduction}
Edge-coloring problem has been one of the foundational problems in graph theory and discrete mathematics since its appearance in 1880 \cite{tait1880remarks}  in connection with the four-color problem. 
In this paper, we study \emph{weighted edge coloring problem} which generalizes both unweighted edge-coloring and classical bin packing problem. 
In the online version of the \emph{weighted edge coloring problem} for general undirected multi-graphs, we are given a graph $G:=(V,E)$. At each instance, a new edge $e=(u,v)$ is revealed, and we need to assign a color to this edge at this instant itself. Colors assigned to edges can not be changed in future. Weight of each edge $\in[0,1]$. Let $m$ be the maximum over all vertices of the minimum number of unit-sized bins needed to pack the weights of the incident edges to that vertex. We need to color the edges while maintaining the condition that the sum of weights of edges incident to a vertex and colored with same color must not exceed 1 at any moment. This is called a \emph{proper} coloring of the graph. Our overall objective is to minimize the number of colors used. 

The problem specially finds prominence in  3-stage Clos networks \cite{clos1953study} which reduces to \emph{weighted edge coloring of bipartite graphs}.  We refer the reader to Correa and
Goemans \cite{correa2007improved} for detailed discussion of this reduction and connection with Clos networks. \\

For a given edge weighted undirected graph $G:=(V,E)$ with all edge weights $\in[0,1]$, consider the following notations which will be followed in this paper-
\begin{itemize}
    \item m := $\max_{v\in V} \left\{\textit{Minimum number of unit sized bins required to pack all the edges incident to $v$}\right\}$
    \item n := $\max_{v\in V} \left\{\textit{Sum of weights of all edges incident to $v$}\right\}$
\end{itemize}

\newpage
For the offline weighted edge coloring problem, Feige and Singh \cite{feige2008edge} proved a $\lceil 2.25n \rceil$ upper bound for bipartite graphs. Later Khan and Singh \cite{khan2015weighted} proved an upper bound of $\lceil 2.2223m \rceil$ for bipartite graphs.  Khan \cite{Khan16c} also showed that $2.2m$ colors are sufficient when all items have weights $> 1/4$.
The current best lower bound for the offline version is $1.25m$ for bipartite graphs. For online weighted edge coloring, Correa and Goemans \cite{correa2007improved} proved an upper bound of $5m$ which improved upon $5.75m$ by Gao and Hwang \cite{gao1997wide}. The best known lower bound for the online version is $3m-2$ by Tsai, Wang, and Hwang \cite{tsai2001lower}. 

\subsection{Our contributions}
We design a new algorithm for this problem and prove that $3.39m+o(m)$ colors are sufficient if the maximum number of neighbors of a vertex over all vertices is $o(m)$. This asymptotically improves upon the previous best $5m$ \cite{correa2007improved}. We achieve this improvement using \emph{HARMONIC$_M$} (with $M=12$) online bin packing algorithm \cite{lee1985simple} as a routine in our main algorithm. Most of the previous results, both in the offline and online settings used $FIRST$-$FIT$ instead. Along with that, we present an instance to show that our analysis is tight for our algorithm.\\
For the offline version of the problem, we prove that $m+1$ colors are sufficient for an undirected simple (no multi-edges between 2 vertices) graph with edge disjoint cycles and $1.693m+12$ colors are sufficient for undirected multi-graph trees. 

\subsection{Related Works}
When all the graph is a bipartite graph with only two vertices then the problem reduces to classical bin packing problem, which is well-studied in both offline \cite{de1981bin, karmarkar1982efficient, hoberg2017logarithmic}   and online setting \cite{BaloghBDEL18, Albers0L19, Albers0L20}.  There are many other related generalizations of bin packing \cite{BansalK14, BansalE016, Galvez0AJ0R20}. We refer the readers to \cite{coffman2013bin, ChristensenKPT17} for a survey on bin packing.
For results on edge-coloring, we refer the readers to \cite{stiebitz2012graph, jensen2011graph}.

\section{Preliminaries}

In our algorithms presented in this paper, we either use the \emph{NEXT-FIT} or the $HARMONIC_M$ online bin packing algorithms to assign colors to the edges. We present these subroutines briefly here along with few results which will be used in our analysis.

\subsection{NEXT-FIT algorithm}
$NEXT$-$FIT$ is an online bin packing algorithm in which items of weight $\in[0,1]$ arrive one by one and we need to fit it into an unit-sized bin at the instant of arriving itself. We fit items of size $\rightarrow [0,1]$ into unit-sized bins, in which at any instant we have a single open/active bin where we can place the incoming items. Also, let $m$ be the optimal number of unit-sized bins required for all the items known at hindsight. Let at an instant, item $i$ appear and we have an open bin $b$. If $i$ fits into $b$, then fit $i$ in $b$. Else, close the bin $b$ and take a new fresh bin as the current open bin. Closed bins are never used again. Place $i$ into the new bin.

\begin{lemma}\label{proof next fit}
$2m-1$ bins are sufficient using $NEXT$-$FIT$.
\end{lemma}

\begin{proof}
Let the current open bin be $b_1$ which already has $c$ weight in it and at this instant a new item $i$ of weight $w_i$ arrives. If a new bin is required for accommodating item $i$, then we must have $c+w_i>1$, else a new bin won't be required. Let the new bin be $b_2$. We then close the current bin $b_1$ and we fit the item $i$ in $b_2$. The final weight in $b_2$ will be $\geq w_i$. So, the sum of weights in $b_1$ and $b_2$ is atleast $c+w_i>1$. In the same way, in general the sum of final weights of any 2 consecutive bins is greater than 1. Now assume on contradiction, that at some instant, there is a requirement for opening the $2m^{th}$ bin, then the total weight of all the items which have already arrived will be greater than $m$, which is a contradiction. 
\end{proof}

\subsection{HARMONIC$_M$ algorithm}
We refer reader to \cite{lee1985simple} to know more about \emph{HARMONIC$_M$} algorithm. Essentially, here we have $M$ types/categories of items based on their size and at any instant we have $M$ open/active bins, one for each type to place the incoming items of that type. If the incoming item belongs to type $i$, then it is placed in the bin designated for type $i$ items. If it can not be placed, then the current bin of type $i$ is closed and a new bin of type $i$ is formed and the item is placed in it. Any closed bins are never used again.

\begin{lemma}\label{proof harmonic}
If $M=12$, then $1.6926m \approx 1.693m$ bins are sufficient \cite{lee1985simple}.

\end{lemma}

\subsection{Definitions}

Now we define some terminologies which we would use in our algorithms and their analysis. From now we always take $M=12$ if we talk about the algorithm which uses \emph{HARMONIC$_M$}. Also, terms like bins and colors are used interchangeably. Coloring an edge with a color is analogous to fitting an item in a bin.

\begin{definition}\label{simple}
(Simple graphs and Multi-graphs): An undirected graph $G$ is called a simple graph if there is atmost 1 edge between any pair of vertices. An undirected multi-graph $G$ can have multiple edges between any pair of vertices. 
\end{definition}

\begin{definition}\label{bins}
(Open, Closed and Empty colors): 
\end{definition}
When using \emph{NEXT-FIT}, we maintain a single active color between any pair of neighbouring vertices at all time. So, if a vertex has $r$ neighbours at an instant, then it will have $r$ active colors, one with each neighbour. This active color is also called the open color between a pair of neighbouring vertices. Let at an instant, an edge $e$ arrive between vertices $u$ and $v$. We first try to color $e$ with the current open color between $u$ and $v$. If it violates the condition of proper coloring, then we close the current open color between $u$ and $v$ and create a new open color between them. If $e$ is the first edge between $u$ and $v$, then we introduce a new open color between them and assign that color to $e$. The colors which are closed in this manner are called closed colors and they are never used again. We can say that all these colors are taken from a set/palette. Our main goal is to provide an upper bound to the sufficient size of the palette. Colors which are neither open nor closed are called empty colors. Empty colors are part of the palette and  act like candidates for being an open color in the future.\\

When using \emph{HARMONIC$_M$}, the definitions are very similar as above. The only change is that now at any instant, between any two neighbouring vertices, we have $M=12$ open colors between them corresponding to each type/category instead of just 1 open color as in \emph{NEXT-FIT}. So, if a vertex has $r$ neighbours at an instant, then it will have $Mr$ active colors, $M$ with each neighbour. Let at an instant, edge $e=(u,v)$ arrive which is a type $i$ edge. We first try to color $e$ with the type $i$ open color between $u$ and $v$. If this violates the condition for proper coloring, then we close the type $i$ color between $u$ and $v$ and introduce a new type $i$ color between $u$ and $v$. If edge $e$ is the first type $i$ edge between $u$ and $v$, then we introduce a type $i$ color between $u$ and $v$ in order to color the edges of type $i$ between $u$ and $v$.

\section{Online Weighted Edge Coloring}

Reiterating the online \emph{Weighted edge coloring problem}, let $G:=(V,E)$ be the underlying undirected multi-graph. Each edge appear/arrive one by one in an online manner. Each edge has weight $w \rightarrow [0,1]$. $m$ is the maximum over all vertices of the minumum number of unit-sized bins needed to pack the
weights of the incident edges to that vertex. At an instant, let edge $e \in E$ of weight $w_e$ arrive. We need to color $e$ such that the invariant/condition that sum of weights of edges incident to any vertex and colored with the same color must not exceed $1$ is maintained at all times. We wish to find an upper bound for the sufficient number of colors required for the above procedure. Essentially, we present $two$ algorithms. One uses \emph{NEXT-FIT} online bin packing procedure in which we show that $4m+2t$ colors are enough, where $t$ is the maximum number of neighbors of a vertex over all vertices. The other algorithm uses \emph{HARMONIC$_M$} online bin packing procedure in which we show that $3.39m+24t$ colors are enough (taking $M=12$), definition of $t$ being the same. So, if $t=o(m)$ we can achieve an asymptotic competitive ratio of $4$ and $3.39$ respectively beating the previous best 5 as mentioned before. In our analysis we take value of $M$ to be $12$. As we will see both the algorithms and their analysis are very similar.

\subsection{Algorithm}

We first present our algorithm which uses $NEXT$-$FIT$. Let $t$ be the maximum number of neighbours of a vertex over all vertices. We have a palette $P$ of $4m+2t$ colors numbered from 1 to $4m+2t$.
\begin{algorithm}
\caption{Algorithm for online weighted edge coloring for undirected multi-graphs using \emph{NEXT-FIT}}
\begin{algorithmic}
  
  \STATE Let at an instant, edge $e=(u,v)$ arrive. If $e$ is the first edge appearing between $u$ and $v$, then choose the least numbered color $c'$ from palette $P$ which is empty for both $u$ and $v$ and designate it as the open color between them. Color $e$ with $c'$. Else, let the current open color between vertices $u$ and $v$ at this instant be $c$.
\end{algorithmic}
\begin{algorithmic}[1]
  
  \STATE If color $c$ is the current open color between $u$ and $v$, and edge $(u,v)$ can be colored with $c$, then
color edge $(u,v)$ with $c$.

  \STATE Else, close color $c$. Find the least numbered color $c'$ from palette $P$ which is an empty color of both $u$ and $v$. It is now the open color between $u $ and $v$. Color $(u,v)$ with $c'$.
\end{algorithmic}
\end{algorithm}

Now we present our algorithm which uses $HARMONIC_M$. Let $t$ be the maximum number of neighbours of a vertex over all vertices. We have a palette $P$ of $3.39m+24t$ colors numbered from 1 to $3.39m+24t$. This algorithm is very similar to the previous algorithm. The only difference is that now we have $M$ types/categories of edges based on its weight which we have to take into account. For each type we have an open color between a pair of neighbouring vertices.

\begin{algorithm}
\caption{Algorithm for online weighted edge coloring for undirected multi-graphs using \emph{HARMONIC$_M$}}
\begin{algorithmic}
  
  \STATE Let at an instant, edge $e=(u,v)$ of type $i$ arrive. If $e$ is the first edge of type $i$ appearing between $u$ and $v$, then choose the least numbered color $c'$ from palette $P$ which is empty for both $u$ and $v$ and designate it as the open color of type $i$ between them. Color $e$ with $c'$. Else, let the current open color of type $i$ between vertices $u$ and $v$ at this instant be $c$.
\end{algorithmic}
\begin{algorithmic}[1]
  
  \STATE If color $c$ is the current open color of type $i$ between $u$ and $v$, and edge $(u,v)$ can be colored with $c$, then
color edge $(u,v)$ with $c$.

  \STATE Else, close color $c$. Find the least numbered color $c'$ from palette $P$ which is an empty color of both $u$ and $v$. It is now the open color of type $i$ between $u $ and $v$. Color $(u,v)$ with $c'$.
\end{algorithmic}
\end{algorithm}

\newpage

In the next section, we would prove that palette size of $4m+2t$ and $3.39m+24t$ is sufficient for the algorithms using $NEXT$-$FIT$ and $HARMONIC_M$ respectively.

\subsection{Analysis}

Consider the following lemmas.

\begin{lemma}\label{nf1}
In an online bin packing instance, if we partition all items into $t$ non-empty types, and introduce a requirement that two items of different types cannot be packed in a single bin, then $2m-1+t$ bins are sufficient using NEXT-FIT. Also, at all times, we have exactly $t$ open bins, thus there can be at most $2m-1$ closed bins at any instant.
\end{lemma}
\begin{proof}
Assign $t$ open bins at the start, named $c_{1}, c_{2}, ..., c_{t}$, where bin $c_{i}$ would accommodate items of type $i$. On the arrival of an item of type $i$, if the current bin $c_{i}$ is unable to fit the item, then close the bin $c_{i}$, and replace it with another new bin, which will now be named $c_{i}$. Place the item in this new bin of type $i$. Carry on this procedure for all items.\\
At last there will be $t$ open bins, and the number of closed bins will be $< 2m$. Otherwise, if number of closed bins is $\geq 2m$, then it would be a contradiction to the $2m-1$ bound guaranteed by Lemma \ref{proof next fit}. This is because here we are essentially following the same rules as standard \emph{NEXT-FIT}, the only difference is upon arrival of an item of type $i$, we are checking the bin $c_i$ instead of a common open bin for all items in standard \emph{NEXT-FIT}. Using the standard \emph{NEXT-FIT}, first produce the items of type 1, then type 2 and so on till type $t$. So, if $\geq 2m$ closed bins are required, then in this case also $\geq 2m$ will be required contradicting Lemma \ref{proof next fit}.  
Thus, overall we need at most $2m-1+t$ different bins for all items under this scenario. In other words, at most $2m-1$ closed bins as we have exactly $t$ open bins.
\end{proof}

\begin{corollary}\label{h1}
Using $HARMONIC_M$ (with $M=12$) in the same scenario of Lemma \ref{nf1}, $1.693m+12t$ bins are sufficient. More specifically, there would be $12t$ open bins at any instant, and there would be at most $1.693m$ closed bins at any instant using Lemma \ref{proof harmonic}.
\end{corollary}

Let $t$ be the maximum number of neighbors of a vertex over all vertices of $G:=(V,E)$. Now we will first present the analysis for the algorithm which uses $NEXT$-$FIT$ and prove our claim that $4m+2t$ colors are sufficient. The bound using $HARMONIC_M$ would easily follow by just replacing $2m$ with $1.693m$ and $t$ with $12t$ in the following analysis.

\newpage
\begin{lemma}\label{l2}
For any vertex $w$ of the graph, the number of closed bins used by $w$ is at most $2m-1$ at any instant.  
\end{lemma}
\begin{proof}
The number of open bins is clearly at most $t$ because $w$ has at most $t$ neighbours. So, specifically, we need to show that number of closed bins used by $w$ is at most $2m-1$. Actually, this follows directly from Lemma \ref{nf1}. We can relate the $t$ types in Lemma \ref{nf1} with $t$ neighbours here, and the items with the edges.

We will prove it by contradiction. Let at an instant, edge $e$ incident to $w$ arrive which causes it's $2m^{th}$ bin to close in order to fit edge $e$. The claim is that such an instant can $never$ occur. Let suppose such an incident indeed occur. Now, take all the items which were previously packed in the $2m$ closed bins and the item $e$. Then, imagine producing all these items (i.e. edges) in the same order as they arrive in the original graph scenario, and try to pack these items using \emph{NEXT-FIT} along with the condition that there are $t$ types (corresponding to $t$ neighbour vertices) and items of 2 different types cannot be placed together. We would then have to close the $2m^{th}$ when the item $e$ arrives leading to $2m$ closed bins, which is a contradiction to Lemma \ref{nf1}. 
\end{proof}

Now we arrive at our main theorem.

\begin{theorem}\label{main1}
For online edge arrival scenario of proper weighted coloring of graphs, $1.693\cdot 2m + 24t$ colors are sufficient using $HARMONIC_M$ (with $M=12$) and $4m + 2t$ colors are sufficient using $NEXT$-$FIT$ algorithm for a proper coloring solution. $t$ is the maximum number of neighbours  over all the vertices of the graph and $m$ is the maximum over all vertices of minimum number of bins required to pack all edges incident to a vertex such that sum of weights packed in each bin $\leq 1$.
\end{theorem}

\begin{proof}
When using \emph{NEXT-FIT}, let at an instant, edge $e=(u,v)$ arrive, 
and let $g_1$ = number of closed bins used by $u$ and  $g_2$ = number of closed bins used by $v$ and let bin $c$ be the current open bin between $u$ and $v$, all just before arrival of edge $e$.

If edge $e$ is able to fit in bin $c$, then fit edge $e$ in bin $c$ and continue.
Else, close bin $c$. Replace it with a new bin and fit edge $e$ in the new bin.

Now, we need to prove that if we have a palette of $4m + 2t$ colors, then the new bin can be found from this palette without any violations. In other words, $4m+2t$ colors are sufficient. From Lemma \ref{l2}, we have $g_1$ $<$ $2m$ and $g_2$ $<$ $2m$. Vertex $u$ has at most $t$ open bins and vertex $v$ has at most $t$ open bins just before arrival of $e$. So, just before arrival of $e$, the sum of closed and open bins for both vertices $u$ and $v$ is $g_1+g_2+2t$ $\leq$ $2m-1+2m-1+2t = 4m-2+2t$ (It is indeed $4m-2+2t-1$ because of the common open bin between $u$ and $v$). So, given our pool of $4m + 2t$ colors, we will always be able to find a new color without any violations to either vertices. Also, $u$ and $v$ are the only vertices where a violation can occur at this step and thus we care about only these.

If a new bin is selected, then $g_1$ and $g_2$ both increments by $1$ but the number of open bins for both remains at most $t$. Though $g_1$ and $g_2$ increase by $1$, it is guaranteed that they never reach $2m$. Because it would be a violation to Lemma \ref{l2} as then $g_1=2m$ $>$ $2m-1$, contradicting Lemma \ref{l2}.

In other words, the above arguments imply that if a specific vertex has $2m-1$ closed bins, then none of it's $t$ open bins will be forced to close. Because a bin is closed only when another new bin in opened and thus leading to a contradiction of Lemma \ref{l2}.

Finally, we can say that $4m+2t$ colors are enough (indeed $4m-1+2t$ colors are enough). The bound $1.693 \cdot 2m-1+24t$ directly follows for $HARMONIC_M$ (with $M=12$) with the exact same arguments as above, just replace $t$ with $12t$ and $2m$ with $1.693m$. Also as described in Algorithm 2, take into account the type of edge arriving. An edge $e=(u,v)$ of type $i$ must be fit into the open bin of type $i$ between $u$ and $v$.
\end{proof}

\begin{corollary}
From Theorem \ref{main1} we get, if $t=o(m)$, we have a competitive ratio of $1.693\cdot 2 < 3.39$ using $HARMONIC_M$ (with $M=12$) and $4$ using $NEXT$-$FIT$.\\
In other words, the new best asymptotic competitive ratio is $3.39$, an improvement over previous best of $5$. 
\end{corollary}

\subsection{Tightness of Analysis}

The number $1.6910$ is the greatest lower bound for the worst-case performance ratio of any $O(1)$-space online bin packing algorithm \cite{lee1985simple}. For $M=12$, the best bound is $\approx 1.6926 < 1.693$ \cite{lee1985simple}. Now, we would be constructing an example where at least $1.693\cdot 2=3.386m \approx 3.39m$ colors are required when using our $HARMONIC_M$ coloring algorithm with $M=12$. In other words, where the competitive ratio is $3.39$. $3.39$ is chosen for simplicity in expressing the procedure, else following the below construction idea, competitive ratio very close to $1.691\cdot 2$ can be achieved by increasing the value of $M$.

Choose positive integers $\epsilon$ and $m$ such that $m\approx10^{20}$, $\frac{m}{\epsilon}\approx10^{10}$, $\frac{\epsilon}{12}\approx10^{10}$ and $1.693m, \frac{1.693m}{\epsilon}$ are integers. One thing to note here is that on the arrival of an edge $e=(u,v)$, if $e$ can not fit in the current open bin between $u$ and $v$, then the algorithm should provide a deterministic way of choosing the next open bin between them. The way is that the algorithm chooses the least numbered bin possible as the next new open bin between $u$ and $v$. Below are the steps of construction of the tight example taking $M=12$. 

\begin{itemize}
\item Let $u_1$ and $v_1$ be the first two vertices. Introduce edges only between them such that the number of closed bins for them be $1.693m-\epsilon$. Choose these bins in a contiguous manner starting from the first bin. This can be achieved because we have the bound equal to $1.693m$ as stated above. Also note that due to the lower bound, the number of closed bins for a vertex can be at least up to $1.691m-1$. (may be even more depending upon the value of $M$).
\item Like this form another $\frac{1.693m}{\epsilon}-1$ pairs of vertices $u_i$ and $v_i$ for $i=2,3,4...,\frac{1.693m}{\epsilon}$. So, now we have $\frac{1.693m}{\epsilon}$ pairs of vertices, with each pair having the first $1.693m-\epsilon$ bins as closed bins between them.
\item Now we introduce a new vertex $w$ with multiple edges between vertices $u_1$ and $w$ such that the number of closed bins needed by $w$ is $\epsilon$. Now $w$ can not choose these $\epsilon$ bins from the first $1.693m-\epsilon$ bins as then it would create a clash with vertex $u_1$. So, $w$ chooses bins from $1.693m-\epsilon$ to $1.693m$. Then, introduce multiple edges between vertices $u_2$ and $w$ such that $\epsilon$ closed bins are required. With similar reasoning as before, $w$ chooses these bins to be from $1.693m$ to $1.693m+\epsilon$. Then, introduce multiple edges between vertices $u_3$ and $w$ such that $\epsilon$ closed bins are required. With similar reasoning as before, $w$ chooses these bins to be from $1.693m+\epsilon$ to $1.693m+2\cdot\epsilon$. Carry on this procedure for all $u_i$ for $i=4,5,...,\frac{1.693m}{\epsilon}$.
\item Thus now vertex $w$ has used bins from $1.693m-\epsilon$ to $1.693m+(\frac{1.693m}{\epsilon}-1)\cdot\epsilon=3.386m-\epsilon$.
So, in total $3.386m-\epsilon$ bins are used in the whole algorithm leading to a competitive ratio of $3.386 \approx 3.39$ (because $m\gg\epsilon$).

Note that here the number of neighbours of $w=\frac{1.693m}{\epsilon}\approx10^{10}\approx m^{0.5}\approx o(m)$. Also, $\epsilon \gg 12$ is chosen so that we don't care about the $12$ open bins between any pair of neighbour vertices in any of the above steps.
\end{itemize}

\section{Offline Weighted Edge Coloring}
In the offline version of the problem, we are given the whole undirected graph upfront. We need to assign colors to each edge maintaining the condition of proper coloring. The current best upper bound is $\lceil 2.2223m \rceil$ given by Khan and Singh \cite{khan2015weighted} for bipartite graphs. We will present improved algorithms with better upper bounds for some special graphs.

\subsection{Simple graph with edge disjoint cycles}
\begin{theorem}
Given an undirected simple graph (no multi-edges between 2 vertices) $G:=(V,E)$ with edge disjoint cycles, it can colored using $m+1$ colors.
\end{theorem} 
\begin{proof}

We have a palette $P$ of $m+1$ colors numbered from 1 to $m+1$. Consider a simpler example where we are given an undirected simple graph $T$ with no cycles that is $T$ is a undirected simple tree. Let vertex $r$ be the root vertex (if not given, choose any vertex as $r$). We will traverse $T$ in a BFS manner starting from the root node $r$. Whenever visiting a node $u$, we will color the edges incident to it.

Consider any step of this traversal where we visit node $u$. Let the $k$ child nodes of $u$ be $v_1,v_2,\dots,v_k$ and let $w$ be the parent node of $u$ (if $u\neq r$). As we are  traversing in a BFS manner, the edge $(w,u)$ must have been already colored and none of the  edges $(u,v_1),(u,v_2),\dots,(u,v_k)$ are colored (which we have to color in this step). So, with respect to $u$, its only one of the incident edges have been colored. So, by the definition of $m$, all the remaining edges incident to $u$ can be colored taking at most $m$ colors in total from palette $P$. This will not create any violations in any of the child nodes, as only one of the edges of each child node is colored in this step. Carry on this procedure for each node visited and eventually we have colored all the edges of the tree. Thus, we can properly color a undirected simple tree in $m$ colors.

Now if we also have edge disjoint cycles, then this case is very similar to the above. The only difference is that when visiting a node $u$, it might happen that two of its edges have already been assigned different colors. But it might happen that the optimal packing of edges incident to $u$ in $m$ unit-sized bins assign same colors to these two already colored edges or vice-versa. It might also happen that when we color the edges between the current vertex and any of its child node, a clash might occur at the child node. This is because now  2 edges of a child node may be colored before visiting that child node which may create a violation. Thus we place an extra color in palette $P$ to overcome this violation. Thus, $m+1$ colors are sufficient in this case.

This result can be further generalised. If we are given an undirected simple graph $G$. Let $y$ be the maximum number of cycles an edge is a part of taken over all the edges. Then while traversing the graph in a BFS manner, it may happen that $y+1$ edges have already been colored. So, in the worst case scenario, we would require $m+y$ colors for properly coloring all the edges in this case.
\end{proof}

\subsection{Multi-graph tree}

\begin{theorem}
Given an undirected multi-graph tree $G:=(V,E)$, it can be properly colored using $2m$ colors using \emph{NEXT-FIT} and using $1.693m + 12$ colors using \emph{HARMONIC$_M$} (with $M=12$). Below we present the algorithms and their analysis.
\end{theorem}

\begin{proof}
Let $r$ be the root vertex (assign any of the vertex as root vertex randomly if $r$ not specified). Below we present the algorithms achieving the above upper bounds along with their analysis. The algorithm starts from the root vertex and visits all the nodes in a BFS manner. When visiting a node, it colors the edges incident to it.

First let us consider when we use \emph{NEXT-FIT}. Let there be a palette $P$ having $2m$ colors numbered from 1 to $2m$. Whenever a new color is required from $P$, we chose the least numbered color which does not create any violations. We claim that the size of $P$ is sufficient to properly color all the edges. When visiting a node $u$. Let the $k$ child nodes of $u$ be $v_1,v_2,\dots, v_k$ and let vertex $w$ be the parent node of $u$ if $u$ is not the root vertex. As we are traversing in a BFS manner, none of the edges between $u$ and its child node $v_i$ will be colored (they have to be colored in this step) whereas the edges between $u$ and parent node $w$ must have already been colored when $w$ was visited.

We will color the edges between $u$ and its child nodes in a specific order in which we consider one child node at a time and color the edges between them. First color the edges between $u$ and $v_1$ in a \emph{NEXT-FIT} manner. This means that at any moment there is an active/open color between $u$ and $v_1$ and present the edges between them in an online fashion (though we have all the edges beforehand) and color them using \emph{NEXT-FIT}. As the edges between $u$ and $w$ have already been colored, let the colors used for them be numbered from $c_1$ to $c_2$. Note that it would a continuous range of colors due to the way we are coloring the edges each child node at a time, and the way of choosing the new color when required from palette $P$. Keep $c_2$ as the initial open color between $u$ and $v_1$ and carry on the procedure as explained above. When all the edges between $u$ and $v_1$ have been colored, let the colors used for this be numbered from $c_2$ to $c_3$. Now move on to the next child node $v_2$ with $c_3$ as the initial open color between them. Carry on the this procedure for all the child nodes of $u$. 

We claim that if the above algorithm is followed, then the number of colors used by $u$ will be atmost $2m$ which are taken from palette $P$ without creating any violations. This can be seen with the help of Lemma \ref{proof next fit}. Lemma \ref{proof next fit} claims that when items are presented in an online manner and we pack them using \emph{NEXT-FIT}, then atmost $2m-1$ bins are required. It is an easy observation that we are doing the same thing here when visiting vertex $u$. The items are the edges and the bins are the colors. As $w$ is the parent node of $u$, $w$ must have been visited earlier than $u$. At that step when considering the edges between $w$ and $u$, following the above algorithm, these edges are presented in an online manner and colored using \emph{NEXT-FIT}. So, when we visit $u$ in the current step, some of its edges have already been colored using \emph{NEXT-FIT}, and now we color the remaining incident edges again using \emph{NEXT-FIT}. Thus, the overall effect with respect to $u$ is that all its incident edges are colored in a \emph{NEXT-FIT} manner thus taking atmost $2m-1+1$ colors. The $+1$ term is there because the first color ($c_1$ in the above example) might not be fully utilized with respect to $u$.  This holds true for any node $u$ visited during BFS traversal. And thus by induction, all the edges of $G$ can be colored properly using only the colors from palette $P$.

Finally for the case when we use $HARMONIC_M$ (with $M=12$), the algorithm and analysis is very similar to the above. The only differences are that now we follow $HARMONIC_M$ packing, by taking into consideration the $M$ types of items/edges. So, at any moment there would be $M=12$ open bins instead of just 1, each corresponding to a unique type. From Lemma \ref{proof harmonic}, the upper bound guarantee is $1.693m$. So, the palette size of $1.693m+12$ will be sufficient for properly coloring all the edges without any violations. Below we present both the algorithms in a concise form.

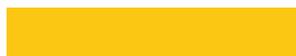
\begin{algorithm}
\caption{Algorithm for proper coloring of multi-graph trees using \emph{NEXT-FIT}}
\begin{algorithmic}
\STATE Let $P$ be a palette/set of $2m$ colors. Graph $G$ will use colors only from the palette $P$. Also, let $Q$ be a FIFO queue. It is empty initially. Insert root vertex $r$ into $Q$. Now repeat the below step until $Q$ is empty.
\end{algorithmic}

\begin{algorithmic}[1]
\STATE Dequeue a vertex from $Q$. Let it be $u$. If $u=r$, then color the edges between $r$ and its child nodes taking one child node at a time as explained above. Take color number 1 as the first open color between $r$ and the first child node.
\STATE Else if $u\neq r$, then let vertex $w$ be the parent of $u$. Let the colors used for coloring the edges between $w$ and $u$ be from $c_1$ to $c_2$. Now color the edges between $u$ and its child nodes taking one child node at a time as explained before. Take $c_2$ as the initial open color between $u$ and the first child node $v_1$.
\STATE Insert into $Q$ all the child nodes of $u$ if any.

\end{algorithmic}
\begin{algorithmic}
\STATE If the given graph has multiple connected components, do the above procedure for each connected component.  
\end{algorithmic}
\end{algorithm}

\begin{algorithm}
\caption{Algorithm for proper coloring of multi-graph trees using \emph{HARMONIC$_M$} with $M=12$}
\begin{algorithmic}
\STATE Let $P$ be a palette/set of $1.693m+12$ colors. Graph $G$ will use colors only from the palette $P$. Also, let $Q$ be a FIFO queue. It is empty initially. Insert root vertex $r$ into $Q$. Now repeat the below step until $Q$ is empty.
\end{algorithmic}
\begin{algorithmic}[1]
\STATE Dequeue a vertex from $Q$. Let it be $u$. If $u=r$, then color the edges between $r$ and its child nodes taking one child node at a time as explained above. Take color number 1 to 12 as the first open colors for each type between $r$ and the first child node.
\STATE Else if $u\neq r$, then let vertex $w$ be the parent of $u$. Let the open colors for each type be $c_1,c_2,\dots,c_{12}$ at the moment when all the edges between $u$ and $w$ were finished assigning colors. Use these same 12 colors as the initial open colors of each type between $u$ and the first child node $v_1$. Carry on the coloring using $HARMONIC_M$ as explained before. 
\STATE Insert into $Q$ all the child nodes of $u$ if any.

\end{algorithmic}
\begin{algorithmic}
\STATE If the given graph has multiple connected components, do the above procedure for each connected component.  
\end{algorithmic}
\end{algorithm}
\end{proof}
 
 \section{Conclusion}
 \begin{itemize}
     \item In all the algorithms above, we have used either \emph{NEXT-FIT} or $HARMONIC_M$. The only differences between them is that there is a single open bin/color in case of \emph{NEXT-FIT} whereas there are 12 open bins/colors in case of $HARMONIC_M$ each for each type of items. Due to this distinction, the upper bound guarantees are different as we saw in Lemma \ref{proof next fit} and Lemma \ref{proof harmonic}. Because of this the algorithms and their analysis when using both are very similar. We just need to take into consideration the type of items when using $HARMONIC_M$ and we replace the upper bound $2m-1$ with $1.693m$ in the analysis along with replacing $1$ open color with $M=12$ open colors at any moment.
     
     \item Note that the \emph{weighted edge coloring problem} with its condition of proper coloring is a kind of global constraint as for a connected component, assigning a color to an edge might depend on or might affect the colors assigned to edges incident to some distant vertex. But in the above algorithms, we just considered any arbitrary step, and we show that there won't be any violations in that step for which we only have to consider the local neighbourhood. Then as this is true for any step, there is no violations occurring in the whole algorithm by induction.
     
     \item In Section 3, we presented better upper bounds for \emph{offline weighted edge coloring} using \emph{online bin packing algorithms} like \emph{NEXT-FIT} and $HARMONIC_M$. Thus, though its an offline problem, still online heuristics can help to come up with better bounds.
 \end{itemize}

\section{Acknowledgment}
The author wants to thank Prof. Arindam Khan, IISc Bangalore, for his continuous support, guidance and regular feedback.
\bibliography{lipics-v2019-sample-article}

\begin{thebibliography}{10}

\bibitem{Albers0L19}
Susanne Albers, Arindam Khan, and Leon Ladewig.
\newblock Improved online algorithms for knapsack and {GAP} in the random order
  model.
\newblock In {\em APPROX/RANDOM}, pages 22:1--22:23, 2019.

\bibitem{Albers0L20}
Susanne Albers, Arindam Khan, and Leon Ladewig.
\newblock Best fit bin packing with random order revisited.
\newblock In {\em MFCS}, pages 7:1--7:15, 2020.

\bibitem{BaloghBDEL18}
J{\'{a}}nos Balogh, J{\'{o}}zsef B{\'{e}}k{\'{e}}si, Gy{\"{o}}rgy D{\'{o}}sa,
  Leah Epstein, and Asaf Levin.
\newblock A new and improved algorithm for online bin packing.
\newblock In {\em ESA}, pages 5:1--5:14, 2018.

\bibitem{BansalE016}
Nikhil Bansal, Marek Eli{\'{a}}s, and Arindam Khan.
\newblock Improved approximation for vector bin packing.
\newblock In {\em SODA}, pages 1561--1579, 2016.

\bibitem{BansalK14}
Nikhil Bansal and Arindam Khan.
\newblock Improved approximation algorithm for two-dimensional bin packing.
\newblock In {\em SODA}, pages 13--25, 2014.

\bibitem{ChristensenKPT17}
Henrik~I. Christensen, Arindam Khan, Sebastian Pokutta, and Prasad Tetali.
\newblock Approximation and online algorithms for multidimensional bin packing:
  {A} survey.
\newblock {\em Computer Science Review}, 24:63--79, 2017.

\bibitem{clos1953study}
Charles Clos.
\newblock A study of non-blocking switching networks.
\newblock {\em Bell System Technical Journal}, 32(2):406--424, 1953.

\bibitem{coffman2013bin}
Edward~G Coffman, J{\'a}nos Csirik, G{\'a}bor Galambos, Silvano Martello, and
  Daniele Vigo.
\newblock Bin packing approximation algorithms: survey and classification.
\newblock In {\em Handbook of combinatorial optimization}, pages 455--531.
  Springer New York, 2013.

\bibitem{correa2007improved}
Jos{\'e}~R Correa and Michel~X Goemans.
\newblock Improved bounds on nonblocking 3-stage clos networks.
\newblock {\em SIAM Journal on Computing}, 37(3):870--894, 2007.

\bibitem{de1981bin}
W~Fernandez De~La~Vega and George~S. Lueker.
\newblock Bin packing can be solved within 1+ $\varepsilon$ in linear time.
\newblock {\em Combinatorica}, 1(4):349--355, 1981.

\bibitem{feige2008edge}
Uriel Feige and Mohit Singh.
\newblock Edge coloring and decompositions of weighted graphs.
\newblock In {\em European Symposium on Algorithms}, pages 405--416. Springer,
  2008.

\bibitem{Galvez0AJ0R20}
Waldo G{\'{a}}lvez, Fabrizio Grandoni, Afrouz~Jabal Ameli, Klaus Jansen,
  Arindam Khan, and Malin Rau.
\newblock A tight (3/2+{\(\epsilon\)}) approximation for skewed strip packing.
\newblock In {\em APPROX/RANDOM}, pages 44:1--44:18, 2020.

\bibitem{gao1997wide}
Biao Gao and Frank~K Hwang.
\newblock Wide-sense nonblocking for multirate 3-stage clos networks.
\newblock {\em Theoretical Computer Science}, 182(1-2):171--182, 1997.

\bibitem{hoberg2017logarithmic}
Rebecca Hoberg and Thomas Rothvoss.
\newblock A logarithmic additive integrality gap for bin packing.
\newblock In {\em SODA}, pages 2616--2625, 2017.

\bibitem{jensen2011graph}
Tommy~R Jensen and Bjarne Toft.
\newblock {\em Graph coloring problems}, volume~39.
\newblock John Wiley \& Sons, 2011.

\bibitem{karmarkar1982efficient}
Narendra Karmarkar and Richard~M Karp.
\newblock An efficient approximation scheme for the one-dimensional bin-packing
  problem.
\newblock In {\em FOCS}, pages 312--320, 1982.

\bibitem{Khan16c}
Arindam Khan.
\newblock {\em Approximation algorithms for multidimensional bin packing}.
\newblock PhD thesis, Georgia Institute of Technology, Atlanta, GA, {USA},
  2016.

\bibitem{khan2015weighted}
Arindam Khan and Mohit Singh.
\newblock On weighted bipartite edge coloring.
\newblock In {\em FSTTCS}, pages 136--150, 2015.

\bibitem{lee1985simple}
Chan~C Lee and Der-Tsai Lee.
\newblock A simple on-line bin-packing algorithm.
\newblock {\em Journal of the ACM (JACM)}, 32(3):562--572, 1985.

\bibitem{stiebitz2012graph}
Michael Stiebitz, Diego Scheide, Bjarne Toft, and Lene~M Favrholdt.
\newblock {\em Graph edge coloring: Vizing's theorem and Goldberg's
  conjecture}, volume~75.
\newblock John Wiley \& Sons, 2012.

\bibitem{tait1880remarks}
Peter~Guthrie Tait.
\newblock Remarks on the colouring of maps.
\newblock In {\em Proc. Roy. Soc. Edinburgh}, volume~10, pages 501--503, 1880.

\bibitem{tsai2001lower}
Kuo-Hui Tsai, Da-Wei Wang, and Frank Hwang.
\newblock Lower bounds for wide-sense nonblocking clos network.
\newblock {\em Theoretical computer science}, 261(2):323--328, 2001.

\end{thebibliography}

\end{document}